\documentclass[10pt,twocolumn,a4paper]{revtex4}

\usepackage{graphicx}
\usepackage{amsmath,amssymb,amsthm,mathtools}
\usepackage{enumerate}
\usepackage{braket}
\usepackage{subfigure}
\usepackage{quantikz}
\usepackage{tikz-cd}
\usepackage[normalem]{ulem}
\usetikzlibrary{arrows.meta,calc,decorations.pathreplacing,fit,positioning,shapes.misc,shapes.symbols}
\usepackage{hyperref}

\usepackage{complexity}

\mathtoolsset{showonlyrefs,showmanualtags}

\newtheorem{theorem}{Theorem}[]
\newtheorem{lemma}[theorem]{Lemma}

\newcommand{\cC}{\mathcal{C}}

\newcommand{\post}{\mathsf{post}}
\newcommand{\prep}{\mathsf{prep}}
\newcommand{\FToutput}{y}

\date{\today}

\begin{document}

\title{Unconditional quantum advantage with \\
noisy planar architectures}

\author{Libor Caha}
\author{Robert Koenig}
\author{Louis Paletta}

\affiliation{Department of Mathematics, Technical University of Munich, 85748 Garching, Germany}
\affiliation{Munich Center for Quantum Science and Technology (MCQST), 80799 M\"unchen, Germany}

\begin{abstract}
We consider  quantum devices restricted to local operations in $2D$ and subject to local stochastic noise below a constant threshold. 
We show that such circuits are computationally  more powerful than $\AC^0$-circuits, i.e., 
noise-free, geometrically-unconstrained constant-depth  classical circuits with unbounded fan-in AND, OR and NOT gates. To this end, we exhibit a computational problem with the following properties: 
(i) Any instance of the problem is correctly  solved with high probability by a 
certain geometrically $2D$-local quantum circuit even if the latter is imperfectly implemented, but 
(ii) any polynomial-size $\AC^0$-circuit fails to solve certain instances of the problem with constant probability. To our knowledge, this is
the first complexity-theoretic separation which applies to planar quantum devices, incorporates noise-resilience and is unconditional, i.e., does not rely on complexity-theoretic assumptions.    This brings the experimental demonstration of an 
unconditional quantum advantage closer to experimental realities.
\end{abstract}
\maketitle

\textit{Introduction.---} In the passage from mathematical abstraction to physical realization, the elegance of theoretical quantum computation meets stubborn physical realities in its experimental realization. Quantum information is continuously eroded by system-environment interactions, crosstalk, and control imperfections that corrupt state preparation, measurement, and unitary operations. Suitable fault-tolerance mechanisms must balance operating on encoded information against maintaining sufficient isolation and redundancy to protect it~\cite{shor1996fault,steane1996error,knill1998resilient,AharonovBenOr}. Compounding this challenge, dominant experimental platforms -- superconducting qubit arrays, quantum optical repeater networks -- confine qubits to $2D$ grids with operations restricted to nearest neighbors~\cite{arute2019quantum}. 
Against these obstacles, a central question emerges: does a purported quantum advantage, such as a computational speedup, ``survive the implementation stack'' -- remaining observable once the resource costs of fault-tolerance and geometric locality are accounted for? 
More concretely, can computations implemented by noisy quantum architectures which are local in $2D$  exhibit an advantage over comparable classical computations? Here we target an unconditional separation~\cite{bravyi2020quantum,caha_3d-local_2026} whose only premise is the validity of quantum
mechanics: we make no complexity-theoretic assumptions.

\textit{Main result.---} Our main result is an affirmative answer to this question: We show that shallow $2D$-local quantum circuits are computationally strictly more powerful than classical shallow circuits, even when the quantum circuit is subject to realistic physical noise.

More precisely, we define a computational problem of size $n$ that can be solved by a family of quantum circuits on $\mathsf{poly}(n)$ qubits arranged on a planar grid. The circuits are $2D$-local and have constant depth $D=O(1)$.
In other words,  they are specified by a unitary $U=L_1\cdots L_{D}$ with each layer~$L_t$ a product of nearest-neighbor two-qubit gates on disjoint neighboring pairs of qubits. Conversely, we prove that this problem is out of reach for  polynomial-size constant-depth classical circuits with unbounded fan-in AND, OR and unary NOT-gates, a circuit class referred to as $\AC^0$. Importantly, this advantage persists in the presence of noise, provided that the noise strength  remains below a constant fault-tolerance threshold. Specifically, the result holds under the general and physically motivated local stochastic noise model~\cite{gottesman2014faulttolerantquantumcomputationconstant}. The latter encompasses simple noise models such as independent depolarizing noise, but also allows to describe  spatial and temporal correlations, provided that error probabilities decay exponentially with error weight.

The computational problem of interest will be a relation problem, specified by a relation, that is, a subset $R_n\subset \{0,1\}^n\times \{0,1\}^n$.
An element $x\in \pi(R_n)$ obtained by 
taking the first component of a pair~$(x,y)\in R_n$ is called an {\em instance} of the problem.  A (classical) algorithm~$\cC$ taking an instance~$x$ as input is said to solve the problem if its output~$\cC(x)\in \{0,1\}^n$ satisfies $(x,\cC(x))\in R_n$. 
We often consider the (average) success probability of an algorithm~$\cC$ solving the problem for a uniformly random instance. Analogous definitions apply to quantum circuits:
 here the classical input~$x$  parametrizes a (classically controlled) unitary circuit $U(x)$ of size polynomial in the input size~$n$. Importantly, the dependence of the circuit $U(x)$ on~$x$ is constrained: we assume a common circuit layout independent of~$x$, where each  individual gate is
 controlled by at most a constant number of bits of~$x$.
The output of such a circuit is defined by the (marginal) distribution of measurement outcomes when 
 starting from the all-zero state, applying $U(x)$, and measuring in the computational basis. We call such a circuit a (classically controlled) prepare-and-measure circuit.

\begin{theorem}[Quantum advantage with noisy shallow planar circuits against unbounded fan-in classical circuits]\label{thm:maintheorem}
There exists a constant threshold $p_{\rm th} >0$ such that the following holds. There exists a  relation problem $R_n^{\mathsf{FT}}\subset \{0,1\}^{n}\times \{0,1\}^{n}$, a set of instances~$S_n$ 
and a constant $\alpha>0$, such that:
\begin{enumerate}[(i)]
\item there exists a uniform $2D$-local classically controlled constant-depth polynomial-size quantum circuit which, under any local stochastic Pauli noise of strength at most $p_{\rm th}$, 
solves the problem for a uniformly random instance from $S_n$ with probability at least $2/3$, but 
\item any (non-uniform) classical circuit composed of unbounded fan-in AND/OR as well as NOT gates (i.e., $\AC^0$), of depth $D$ solving the problem with average success probability at least $1/3$ has size at least $\exp(n^{1/(\alpha  D)})$.
\end{enumerate}
\end{theorem}
\noindent We note that the scalars $1/3,2/3$ can be replaced by any pair of constants $0<\mu<\nu<1$ by amplification.

Theorem~\ref{thm:maintheorem} establishes a trade-off bound between circuit size and depth for classical circuits with unbounded fan-in AND/OR as well as NOT gates solving the relation problem of interest. In particular, it establishes an unconditional complexity-theoretic separation between the power of  $\AC^0$-circuits and noisy shallow $2D$-local quantum circuits.

The novelty of Theorem~\ref{thm:maintheorem} lies in the fact that the quantum advantage is manifested by a $2D$-local quantum circuit, providing an experimentally realistic circuit structure compatible with planar quantum computing architectures. In contrast, earlier proposals demonstrating an advantage of noisy shallow circuits (both over $\mathsf{NC}^0$- and  $\AC^0$-circuits) required at least $3D$-local quantum circuits~\cite{bravyi2020quantum,caha_3d-local_2026} or even circuits with non-local gates~\cite{grier2021interactivequantumadvantagenoisy}

We start by a brief overview of the construction. The starting point for our quantum advantage proposal (expressed by  Theorem~\ref{thm:maintheorem}) is a  $1D$-local constant-depth quantum circuit~$\mathsf{C}_{\mathsf{1D}}$ which exhibits a quantum advantage against  $\AC^0$-circuits. An example of this kind of circuit was given in  Ref.~\cite{caha_3d-local_2026,caha2024single}; the circuit is $1D$-local, but not fault-tolerant. (This result is a strengthening of earlier work establishing separations against $\NC^0$-circuits~\cite{BGK}, and improves the locality of the quantum circuit from $2D$~\cite{watts2019exponential} to~$1D$.)

The next step is to realize the circuit~$\mathsf{C}_{\mathsf{1D}}$ in a fault-tolerant manner.  To this end, we use the 
construction of~\cite{harleykoenig2026fault}: given a $1D$-local Clifford circuit~$\mathsf{C}_{\mathsf{1D}}$, it produces a fault-tolerant constant-depth circuit~$\mathsf{C}^{\mathsf{FT}}_{\mathsf{2D}}$ realizing it using $2D$-local operations. In more detail, its measurement outcomes can be post-processed by efficient classical computation to reproduce the behavior of the ideal circuit~$\mathsf{C}_{\mathsf{1D}}$; in the presence of noise, this simulation is approximate but with a controlled error in variation distance.

The constructed circuit~$\mathsf{C}^{\mathsf{FT}}_{\mathsf{2D}}$
is fault-tolerant, $2D$-local and solves the original relation problem when combined with classical post-processing. However,
although efficiently (i.e., polynomial-time) computable, this post-processing step  potentially adds additional power to the quantum circuit. This means it cannot be ignored when trying to establish a quantum advantage. 

To address this issue, we follow the 
work of Ref.~\cite{bravyi2020quantum} folding the postprocessing into the definition of a new (derived) computational problem. Roughly, this is a new relation problem where for a given input, a valid output is one for which applying the postprocessing map results in a solution to the original problem.

By definition, the constructed $2D$-local quantum circuit~$\mathsf{C}^{\mathsf{FT}}_{\mathsf{2D}}$ (with postprocessing omitted) fault-tolerantly solves the derived relation problem. It remains to show that this derived problem -- like the original relation problem -- is also computationally hard for $\AC^0$-circuits. To this end, we analyze the input-output dependencies of the postprocessing map and show that it preserves $\AC^0$-hardness.

The precise fault-tolerance construction and the proof of the main theorem follow.

\textit{2D-local fault-tolerance in constant depth.---}
Consider a prepare-and-measure  circuit~$\mathsf{C}_{\mathsf{1D}}$ composed of geometrically $1D$-local (classically controlled) Clifford gates. The following construction provides a
fault-tolerant, $2D$-local constant-depth Clifford circuit  which -- when combined with efficient classical postprocessing -- emulates~$\mathsf{C}_{\mathsf{1D}}$.

\begin{lemma}\label{lem:2DFT}
There exists a constant threshold $p_{\rm th}$ such that the following holds. 
Consider a $1D$-local prepare-and-measure Clifford circuit~$\mathsf{C}_{\mathsf{1D}}$ of depth $D$, parametrized by a string~$x\in \{0,1\}^n$ acting on $n_q(n)$ qubits. For each input $x\in \{0,1\}^n$, let $p(\cdot |x)$ denote the corresponding output distribution on $\{0,1\}^{n_q(n)}$.

Then, for every constant $\varepsilon>0$, there exists a $2D$-local prepare-and-measure Clifford circuit~$\mathsf{C}^{\mathsf{FT}}_{\mathsf{2D}}$ 
of constant depth on $n_q'(n)=n_q(n) D\cdot \mathsf{poly}(\log (n_q(n)D))$ qubits,
 and a 
map $\post: \{0,1\}^{n}\times\{0,1\}^{n_q'(n)} \to \{0,1\}^{n_q(n)}$,
with the following properties.
\begin{enumerate}[(i)]
\item \label{it:mainone}
Consider  a noisy implementation of~$\mathsf{C}^{\mathsf{FT}}_{\mathsf{2D}}$ with under any stochastic Pauli noise of strength at most $p_{\rm th}$. For an input~$x$, let $y\in \{0,1\}^{n_q'(n)}$~denote the measurement outcomes obtained from a noisy implementation. Then the distribution of the post-processed output $\post(x,\FToutput)\in \{0,1\}^{n_q(n)}$ is within total variation distance at most~$\varepsilon$ from $p(\cdot |x)$. 
\item \label{it:maintwo}
The map $\post$ can be implemented by a non-uniform $\AC^0$-circuit of size $\exp(D^2 \poly(\log (n_q(n) D)))$.
\end{enumerate}
\end{lemma}

\begin{figure}
\centering
\includegraphics[width=\columnwidth]{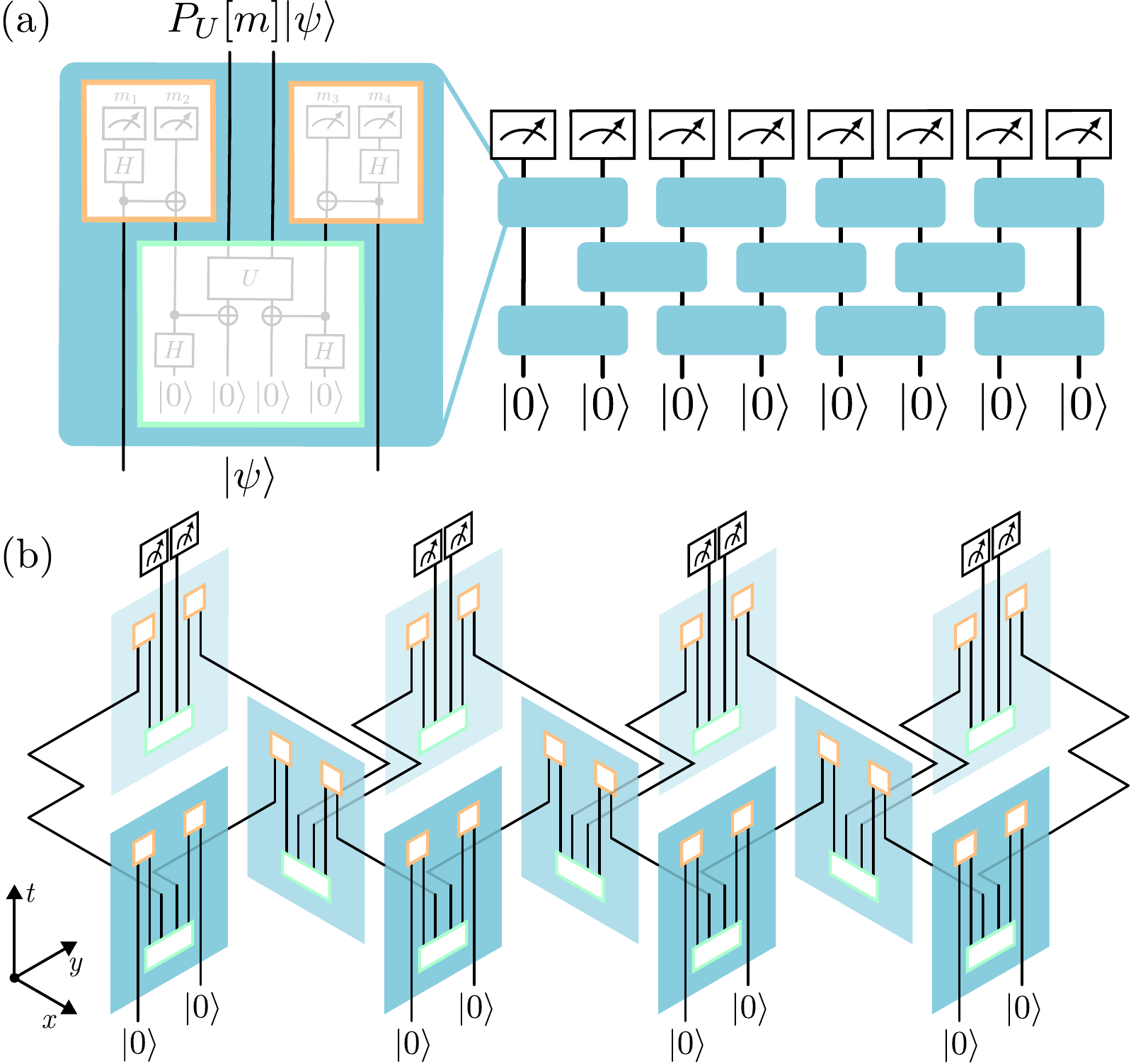}
\caption{Parallelization of a $1D$-local circuit $\mathsf{C}_{\mathsf{1D},*}^{\mathsf{FT}}$: (a) a 2-qubit gate $U$ is recompiled by a gate-teleportation gadget that implements the gate up to a Pauli correction $P_U[m]$, determined by the measurement outcome $m$, (b) because the input to the Bell measurements (shown in orange) is causally after the output of the recompiled unitary, the initial $1D$-local circuit can be unfolded in the $(x,y)$ plane by replacing qubit reinitializations with new qubits in space, and
connecting the output of one gate-teleportation gadget to the input of another. The resulting $2D$-local circuit $\mathsf{C}^{\mathsf{FT}}_{\mathsf{2D}}$ has constant depth corresponding to the gate-teleportation gadget circuit, its initial state preparation and measurements. Wires in the $(x,y)$ plane can equivalently be implemented using qubits on a $2D$ lattice and gates acting over bounded distances.
}
\label{fig:unravelling}
\end{figure}

We remark that the postprocessing map $\post$ in Lemma~\ref{lem:2DFT} can be computed deterministically in polynomial time.
\begin{proof}
The fault-tolerance compilation follows building blocks from~\cite{harleykoenig2026fault}.
We first apply the $1D$ fault-tolerance construction described in~\cite[Section~5.4]{harleykoenig2026fault}, obtained from the bilinear-array scheme of~\cite{stephens2007universal} (see also~\cite{AharonovBenOr}). It takes as input
a $1D$-local Clifford circuit~$\mathsf{C}_{\mathsf{1D}}$
and produces a recompiled $1D$-local circuit~$\mathsf{C}^{\mathsf{FT}}_{\mathsf{1D}}$ with certain fault-tolerance properties. In more detail, the construction uses $L$ (to be chosen later) levels of concatenation of the Steane code~\cite{steane1996error}, with physical operations restricted to single- and two-qubit operations on a line. 
Each qubit of the ideal Clifford circuit~$\mathsf{C}_{\mathsf{1D}}$ is encoded in a block of $7^L$ physical data qubits, and each ideal circuit location of~$\mathsf{C}_{\mathsf{1D}}$ (an initial state preparation, gate, or measurement) is replaced by its corresponding fault-tolerant gadget. The error-correction steps within each gadget use ancillary measurements and apply Pauli corrections determined solely by measurement outcomes within that gadget. Each fault-tolerant gadget contains at most $\exp(O(L))$ physical locations. The resulting circuit~$\mathsf{C}^{\mathsf{FT}}_{\mathsf{1D}}$ fault-tolerantly implements~$\mathsf{C}_{\mathsf{1D}}$ by $1D$-local operations, at the cost of an exponential depth- and qubit-overhead in~$L$.
It is adaptive, i.e., uses mid-circuit measurement results to control intermediate Pauli correction steps inside individual fault-tolerance gadgets.

We then apply the $1D$-to-$2D$ compilation of~\cite[Sections~7.2--7.3]{harleykoenig2026fault} to the circuit~$\mathsf{C}^{\mathsf{FT}}_{\mathsf{1D}}$. This construction uses gate teleportation~\cite{gottesman_demonstrating_1999,terhal2004adaptive}  to execute Cliffords of  the (adaptive) $1D$-local  circuit~$\mathsf{C}^{\mathsf{FT}}_{\mathsf{1D}}$ by unfolding. The result is a $1D$-local circuit~$\mathsf{C}_{\mathsf{1D},*}^{\mathsf{FT}}$
with comparable depth, size and fault-tolerance properties  as~$\mathsf{C}^{\mathsf{FT}}_{\mathsf{1D}}$, but the additional property that every qubit is active only for short (constant) time periods before being measured.

In a next step, we parallelize the circuit~$\mathsf{C}_{\mathsf{1D},*}^{\mathsf{FT}}$ by unfolding its time direction into a second spatial dimension. 
This results in a non-adaptive $2D$-local constant-depth circuit~$\mathsf{C}^{\mathsf{FT}}_{\mathsf{2D}}$ which implements~$\mathsf{C}_{\mathsf{1D},*}^{\mathsf{FT}}$ up to efficient classical post-processing. To construct this circuit, we first remove the
adaptivity (originating both from the gate-teleportation gadgets and the quantum error-correction cycles): this is achieved  by replacing the adaptive Pauli correction gates by identities (their removal will be accounted for subsequently by a postprocessing function, which  propagates the corresponding Pauli frame "in software" to the measurement results). We then unfold time in space: whenever
the circuit~$\mathsf{C}_{\mathsf{1D},*}^{\mathsf{FT}}$ reinitializes a qubit that has been measured out, we replace that qubit by a newly introduced qubit; this leads to a $2D$ array of qubits.
 This unfolding of the  time direction in space as illustrated in Fig.~\ref{fig:unravelling}. 

 In summary, we have obtained a 
 circuit~$\mathsf{C}^{\mathsf{FT}}_{\mathsf{2D}}$ which is  non-adaptive, $2D$-local and has constant depth. Since the circuit is (a classically controlled) Clifford, all Pauli corrections previously removed can nevertheless still be determined from the output of the circuit, and commuted to a correction on the final physical measurements. Composing these maps, this enables recovering the output of the initial circuit~$\mathsf{C}_{\mathsf{1D}}$. What is gained in this process is fault-tolerance and constant depth;  the price is $2D$- instead of $1D$-locality, and an overhead in qubits determined by the depth of the original circuit~$\mathsf{C}_{\mathsf{1D}}$.

It remains to specify the postprocessing map. Let $z=(z_j)_{j\in[n_q]}$ collect the final physical computational basis readouts of
the output blocks with $z_j\in\{0,1\}^{7^L}$. Here and below, we suppress the dependence of~$n_q$ on~$n$. Let $r$ denote the collection of all remaining
measurement outcomes (obtained within fault-tolerance gadgets and mid-circuit gate teleportation steps). We write $\FToutput=(r,z)$ for the complete measurement record. From $(x,r)$, we compute the accumulated Pauli frame by propagating the omitted corrections in the causal order inherited from the $1D$-local circuit~~$\mathsf{C}_{\mathsf{1D}}$, conjugating them by Clifford gates and adjusting each measurement outcome according to the current Pauli frame when incorporating subsequent corrections.
Let $a_j(x,r)$ denote the $X$-component of the accumulated frame on output block $j\in[n_q]$. We define the postprocessing map pointwise as
\begin{align}
    \big[\post(x,(r,z))\big]_j &:= \mathsf{Read}^{(L)}\bigl(a_j(x,r) \oplus z_j\bigr) \ ,\label{eq:postj}
\end{align}
where $\mathsf{Read}^{(L)}$ is the logical readout map for the level-$L$ concatenated Steane code. This is the post-processing map that needs to be applied to the physical single-qubit measurement results in order to fault-tolerantly emulate a logical $Z$ measurement.  The accumulated frame $a_j(x,r)$ and the readout map~$\mathsf{Read}^{(L)}$ can be evaluated in time polynomial in the circuit size of the compiled circuit~$\mathsf{C}^{\mathsf{FT}}_{\mathsf{2D}}$.

As argued in~\cite[Theorem~5.4.1]{harleykoenig2026fault}, the construction of~$\mathsf{C}^{\mathsf{FT}}_{\mathsf{2D}}$ ensures that the recovered output obtained in this way inherits (up to a controlled fidelity loss) the noise-tolerance guarantees of the fault-tolerance construction 
given by the circuit~$\mathsf{C}^{\mathsf{FT}}_{\mathsf{1D}}$. This is because all applied transformations are local in space-time, implying that error propagation is limited.  More precisely, there is a constant threshold $p_{\rm th}>0$ such that, for every constant $\varepsilon>0$ one can choose $L=O(\log\log (Dn_q/\varepsilon))$ so that the following holds. For every  input $x$ and every local stochastic Pauli noise of strength below $p_{\rm th}$, the distribution of $\post(x,\FToutput)$ is within $\varepsilon$ of the ideal distribution $p(\cdot|x)$ in total variation distance. With this choice of $L$, the $\Theta(n_qD)$  circuit locations of the ideal circuit~$\mathsf{C}_{\mathsf{1D}}$
yield a compiled circuit on at most $n_qD\exp(O(L))=n_qD \poly(\log(n_qD))$ qubits for fixed $\varepsilon$.

Next, we characterize a relevant locality property  of the postprocessing map $\post$. 
We show that this function has the property 
that each output bit depends on at most $k$ input bits. We call this property $k$-locality. It implies that the function can be implemented in  a nonuniform way by an $\AC^0$-circuit of bounded size: Each output bit can be written as a Boolean formula in disjunctive normal form  from its truth table. This formula is a constant depth $\AC^0$-circuit of size~$2^{O(k)}$. In our application, the postprocessing function has $n_q$~output bits; evaluating all output bits in parallel leads to a constant depth $\AC^0$-circuit of size $O(n_q 2^k)$.

To analyze the locality of the postprocessing map, fix a logical output bit $y_j=[\mathsf{post}(x,(r,z))]_j$. We associate each recompiled gadget, including each final logical-readout gadget, with its corresponding location in the original ideal $1D$ circuit~$\mathsf{C}_{\mathsf{1D}}$. We then trace the dependencies of $y_j$ backward through the associated local recovery updates, following these  locations in the ideal circuit~$\mathsf{C}_{\mathsf{1D}}$ in reverse temporal order. We call the circuit locations of~$\mathsf{C}_{\mathsf{1D}}$ visited in manner the (logical) backward lightcone of~$y_j$, illustrated in Fig.~\ref{fig:lightcone}, and denote it by $\mathcal{L}^{\leftarrow}(y_j)$. 
Because $\mathsf{C}_{\mathsf{1D}}$ is a $1D$-local depth-$D$ circuit,
the size~$|\mathcal{L}^{\leftarrow}(y_j)|$ of such a logical backward lightcone is bounded by~$O(D^2)$.

The size of the logical backward lightcone allows us to give an upper bound on the total number of  measurement results of $(r,z)$ and input bits $x$ the  output bit~$y_j$ can depend on. The key point is that the level-$L$ implementation of each fault-tolerance gadget (including state preparation, unitaries and readout)
produces at most $\exp(O(L))$~bits of measurement output, and only depends on an $O(1)$ number input bits (of~$x$) by assumption on~$\mathsf{C}_{\mathsf{1D}}$. In summary, we conclude the post-processing function in Eq.~\eqref{eq:postj}  is $k$-local with $k$ bounded by
\begin{align}
    k&\leq \left(\max_{j\in [n_q]}|\mathcal{L}^{\leftarrow}(y_j)|\right)\cdot  \exp(O(L))\nonumber \\
    &= D^2\poly(\log(n_qD))\nonumber .
\end{align}
\begin{figure}
\centering
\includegraphics[width=\columnwidth]{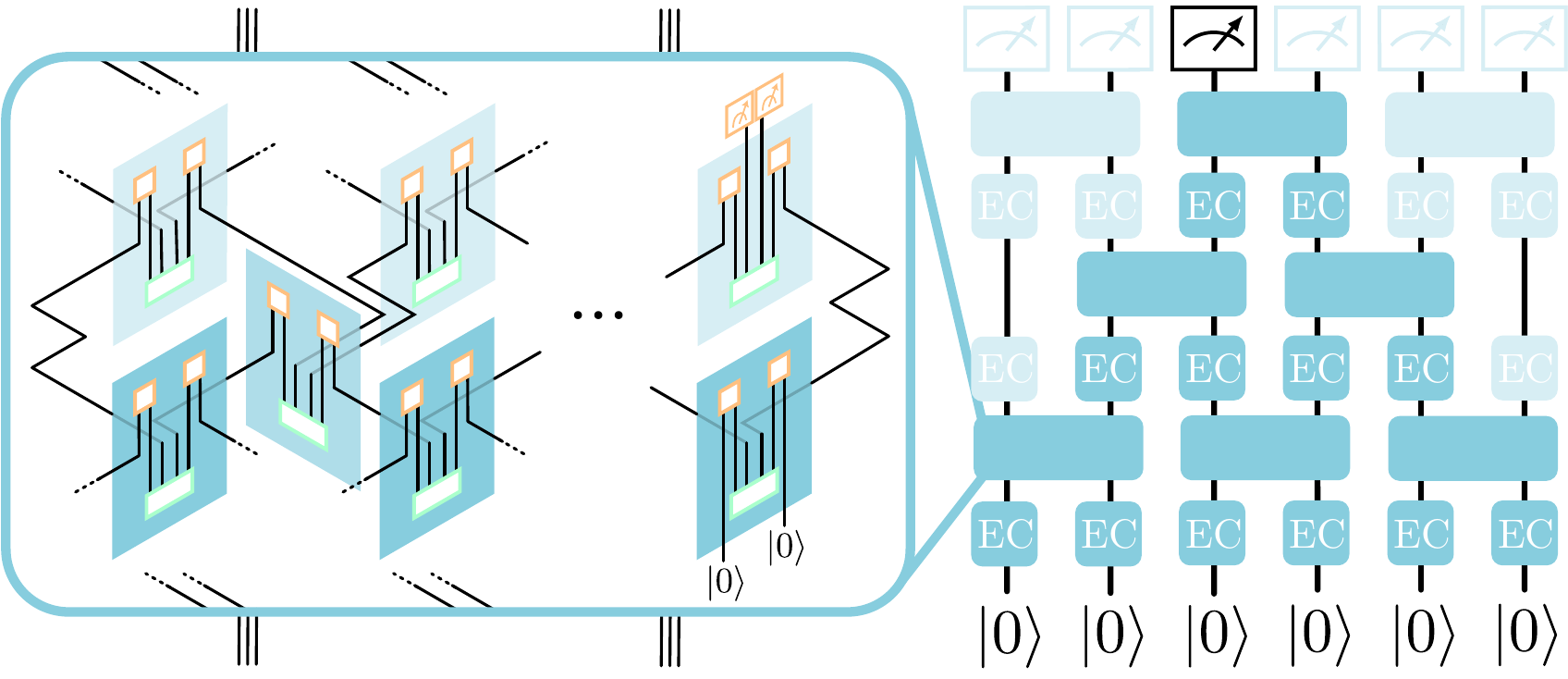}
\caption{Locality of the postprocessing map. Each output block logical measurement is obtained in a way that only depends on its backward lightcone  (shown in dark cyan) — that is, on the corresponding region of the underlying $1D$-local circuit $\mathsf{C}_{\mathsf{1D}}$. A fault-tolerant gadget within that region contributes two types of outcomes (shown in orange): the outcomes of the physical teleportation-based gadgets, and the outcomes of the stabilizer measurements. Each gate in fault-tolerant gadget can also depend on constantly many input bits. The locality of the postprocessing map can thus be bounded by the size of the logical backward lightcone multiplied by the maximal number of measurement outcomes produced and inputs used within any single fault-tolerant gadget.}
\label{fig:lightcone}
\end{figure}
We can thus implement this map by a non-uniform $\AC^0$-circuit of size at most $\exp(D^2\poly(\log(n_qD)))$, where we absorbed the prefactor $n_q$ to the exponent. This is the claim.
\end{proof}
\enlargethispage{2em}
\textit{Quantum advantage with noisy shallow circuits.---}
We are now ready to prove Theorem~\ref{thm:maintheorem}. 
We follow the construction outlined above. We use the so-called single-qubit gate teleportation problem (or its amplified version with its $1D$ circuit flattened from a ring to a line) 
to instantiate our construction. The latter  was shown to be hard for  $\AC^0$-circuits~\cite{caha_3d-local_2026,caha2024single}, while solvable with certainty by a $1D$-local classically controlled Clifford circuit~$\mathsf{C}_{\mathsf{1D}}^{\mathsf{Telep}_n}$. We apply Lemma~\ref{lem:2DFT} to this circuit. Following Ref.~\cite{bravyi2020quantum}, we fold
the classical postprocessing map into the definition of a fault-tolerant relation. For each input, an output string will satisfy the new relation precisely when its recovered output satisfies the original relation.

\begin{proof}[Proof of Theorem~\ref{thm:maintheorem}] 
The result of Ref.~\cite{caha_3d-local_2026} gives a relation $R_n$ on $n$-bit inputs that is solved with certainty by a polynomial-size, constant-depth, classically controlled $1D$-local Clifford circuit. On the other hand, there is a constant~$\beta>0$ such that for all sufficiently large $n$, any depth-$d$ $\AC^0$-circuit solving $R_n$ with average success probability at least $1/3$ on uniformly random inputs has size at least $\exp(n^{1/(\beta d)})$.

We first apply Lemma~\ref{lem:2DFT} with $\varepsilon=1/3$ to
the circuit $\mathsf{C}_{\mathsf{1D}}^{\mathsf{Telep}_n}$. This yields a $2D$-local constant-depth prepare-and-measure circuit~$\mathsf{C}_{\mathsf{2D}}^{\mathsf{FT},\mathsf{Telep}_n}$ together with a postprocessing map $\post$. 

To avoid classical fan-out within the circuit, we introduce a preprocessing map $\prep$ which is an $\AC^0$-circuit. In our case it simply supplies each classically controlled physical gate in the fault-tolerant gadgets with its own copies of the required input bits. With this, we can define the derived fault-tolerant relation~$R_n^{\mathsf{FT}}$  as follows.
We declare $(\prep(x),\FToutput)\in R_n^{\mathsf{FT}}$ if and only if $(x,\post(x,\FToutput))\in R_n$. Note that both relations $R_n$ and $R_n^{\mathsf{FT}}$ are not defined for all integers $n>0$ bot for infinitely many of them.

Since the ideal quantum circuit~$\mathsf{C}_{\mathsf{1D}}^{\mathsf{Telep}_n}$ solves $R_{n}$ with certainty, the total variation bound implies that, under any local stochastic Pauli noise of strength at most $p_{\rm th}$, the compiled circuit~$\mathsf{C}_{\mathsf{2D}}^{\mathsf{FT},\mathsf{Telep}_n}$ solves $R_n^{\mathsf{FT}}$ with probability at least $2/3$  for every input from $\pi(R_n)$. This proves~\ref{it:mainone}.

Suppose a depth-$D$, size-$s$ $\AC^0$-circuit solves $R_n^{\mathsf{FT}}$ with average success probability at least $1/3$. Composing it with the preprocessing map $\prep$ and the postprocessing map $\post$ of Lemma~\ref{lem:2DFT} gives an $\AC^0$-circuit for $R_n$ with the same average success probability, depth $D+O(1)$ (since both $\prep$ and $\post$ are of constant depth), and size $s+\poly(n)+\exp(\poly\log(n))$
(where we added the sizes of $\prep$ and $\post$). 
The cited classical hardness result from~\cite{caha_3d-local_2026} therefore implies $s\ge\exp(n^{1/(\alpha D)})$ for a sufficiently large constant $\alpha$, every fixed $D$, and sufficiently large $n$. This proves~\ref{it:maintwo}.
\end{proof}

\textit{Discussion.---} 
We have identified an example of a computational  task which is solvable by noisy shallow quantum circuits with nearest-neighbor gates on a plane, but is beyond the computational reach of classical $\AC^0$-circuits. This noisy planar setting is of particular practical relevance as it matches the $2D$-connectivity of leading hardware platforms such as superconducting qubits. The established complexity-theoretic separation is unconditional but modest: it does not provide insight into the (potential) differences between efficient (i.e., polynomial-time) classical and quantum computation.
However, our proposal comes with other attractive features. It incorporates noise-tolerance without requiring a scalable fault-tolerant quantum computer. It also provides a means of validating quantum hardware because the correctness of the quantum circuit output can be verified efficiently.
In contrast, commonly considered sampling problems in the area of quantum supremacy~\cite{aaronson2011computational,bouland2019complexity,bremner2016average}
require exponentially many samples to verify~\cite{hangleiter2019sample},  and rely on complexity-theoretic assumptions. Furthermore, unless equipped with fault-tolerance mechanisms~\cite{paletta2024robust,hangleiter2025fault}, several of them are known to become classically simulable under noise, precluding the possibility of a significant computational advantage~\cite{aharonov2023polynomial,schuster2025polynomial,rajakumar2025polynomial}.

It is natural to wonder whether the experimental requirements to observe a quantum advantage can be reduced further beyond our proposal, e.g., using a linear array of qubits. In concurrent work~\cite{toappear}, it is shown that an advantage of noisy geometrically local quantum circuits akin to the one expressed Theorem~\ref{thm:maintheorem} against $\AC^0$-circuits is impossible in $1D$. This means that our proposal for a quantum advantage with noisy shallow circuits in the plane, i.e., in $2D$, is optimal in terms of spatial dimensionality.

Another direction of study concerns quantum advantages for other tasks. In particular it is currently not known whether unconditional advantage in sampling probability distribution~\cite{grier2025quantumadvantagesamplingshallow,watts2026unconditional} can be made fault-tolerant, even allowing for non-local (quantum) connectivity. Efficient verifiability is also a concern for such tasks.

\noindent
\textit{Acknowledgment.---} RK thanks Dylan Harley for useful discussions at an early stage of this project.
LC, RK, and LP acknowledge funding from the European Research Council under Grant Agreement No.~101001976 (project EQUIPTNT) and from the Munich Quantum Valley, which is supported by the Bavarian state government with funds from the Hightech Agenda Bayern Plus. RK would like to thank the Isaac Newton Institute for Mathematical Sciences, Cambridge, for support and hospitality during the programme ``Mathematics of many-body entanglement''
where work on this paper was undertaken; this work was supported by EPSRC grant no EP/Z000580/1.

\noindent
\textit{AI statement.---} The main conceptual contributions were developed by the authors. Claude Fable 5 (Anthropic), GPT-5.6 Sol and GPT-6 Astra (OpenAI) were used in the writing process.

\bibliography{q}
\end{document}